\documentclass[10pt,journal]{IEEEtran}
\normalsize
\usepackage{multirow,amsthm,amssymb,amsmath,bbm,tikz,graphics,cite,float,graphicx,epstopdf,epsfig,verbatim,url,caption,arydshln}
\usepackage{algorithm,algpseudocode,algorithmicx}
\usetikzlibrary{automata}
\usetikzlibrary{shapes,arrows}
\newtheorem{lem}{Lemma}
\newtheorem{corol}{Corollary}
\usepackage[utf8]{inputenc}
\newtheorem{thm}{Theorem}

\newtheorem{rem}{Remark}

\newtheorem{mydef}{Definition}

\usepackage{capt-of}
\allowdisplaybreaks

\makeatletter
\def\blfootnote{\xdef\@thefnmark{}\@footnotetext}
\makeatother

\def\ttx{0.45}

\tikzset{
  treenode/.style = {align=center, inner sep=0pt, text centered,
    font=\sffamily},
  arn_r/.style = {treenode, circle, black, draw=black, 
    text width=1.5em, very thick},
}

\begin{document}
\title{A Secure Approach for Caching Contents in Wireless Ad Hoc Networks}
\author{Mohsen~Karimzadeh~Kiskani$^{\dag}$ and Hamid~R.~Sadjadpour$^{\dag}$}
\maketitle \thispagestyle{empty}
\begin{abstract}
Caching aims to store data locally in some nodes within the network to be able to retrieve the contents in shorter time periods. However, caching in the network did not always consider secure storage (due to the compromise between time performance and security). In this paper, a novel decentralized secure coded caching approach is proposed. In this solution, nodes only transmit coded files to avoid eavesdropper wiretappings and protect the user contents. In this technique random vectors are used to combine the contents using XOR operation. We modeled the proposed coded caching scheme by a Shannon cipher system to show that coded caching achieves asymptotic perfect secrecy. The proposed coded caching scheme significantly simplifies the routing protocol in cached networks while it reduces overcaching and achieves a higher throughput capacity compared to uncoded caching in reactive routing. It is shown that with the proposed coded caching scheme  any content can be retrieved by selecting a random path while achieving asymptotic optimum solution. We have also studied the cache hit probability and shown that the coded cache hit probability is significantly higher than uncoded caching. A secure caching update algorithm is also presented.
\end{abstract}


\IEEEpeerreviewmaketitle

\section{Introduction}
\label{intro_sec}
\blfootnote{Copyright (c) 2015 IEEE. Personal use of this material is permitted. However, permission to use this material for any other purposes must be obtained from the IEEE by sending a request to pubs-permissions@ieee.org.} 
\blfootnote{M. K. Kiskani$^{\dag}$ and H. R. Sadjadpour$^{\dag}$ are with the Department of Electrical Engineering, University of California, Santa Cruz. Email: \{mohsen, hamid\}@soe.ucsc.edu}

Significant advances in wireless and mobile technologies over the past decades, have made it possible for mobile users to stream high quality videos and to download large-sized contents. Video streaming applications like Netflix, Hulu and Amazon have become increasingly popular among mobile users. Close to half of all video plays were on mobile devices like tablets and smartphones\footnote{http://go.ooyala.com/rs/447-EQK-225/images/Ooyala-Global-Video-Index-Q4-2015.pdf} during the fourth quarter of 2015. 

In parallel, storage capacity of mobile devices has significantly increased without being fully utilized. 
Efficient use of this under-utilized storage  is specially important for video streaming applications that account for a large portion of the overall internet traffic.

On the other hand, {\em proprietary contents} as opposed to {\em user-generated contents} are the most important assets of many content sharing companies such as Netflix, Hulu, and Amazon. These companies use extreme measures to protect their data and therefore, they are hesitant to let the end users cache their contents locally. One possible solution is to encrypt each content before caching it locally. Encryption algorithms {reduce} the content sharing rate. Further, such algorithms are only {\em computationally secure} and an adversary is able to break them with time. For instance, Data Encryption Standard (DES) which was once the official Federal Information Processing Standard (FIPS) in US is no longer considered secure. In this paper, we  propose an {\em information theoretically} secure solution for caching contents which cannot be decoded with time. 



 

Physical layer security in wireless networks has been the subject of many recent research papers. With a focus on physical layer security in wireless ad hoc networks, we propose a novel decentralized coded caching approach in which random vectors are used to combine the contents using XOR operation. Nodes only transmit coded files and therefore an eavesdropper with a noiseless channel will not be able to decode the desired contents. Our technique is quite different from techniques which exploit wireless channel dynamics to achieve physical layer security. We demonstrate that coded caching technique is similar to Shannon cipher \cite{shannon1949communication} problem and it can  achieve asymptotic perfect secrecy during file transmissions by using a secure low bandwidth channel to exchange the decoding gains.  
The main contributions of this paper are introduction of decentralized coded caching for wireless ad hoc networks, capacity computation of coded and uncoded caching for proactive and reactive routing strategies\footnote{Notice that the notion of proactive (or reactive) routing that we study is the extreme case when complete (or no) network knowledge is available.}, proof of security,  computation of cache hit probability, and introduction of update algorithm. 

A shorter version of this paper was presented in \cite{DBLP:conf/iccnc/KiskaniS17} that does not include proofs for {Theorems} \ref{thm_capacity} and \ref{thm_coded} and {Lemma} \ref{lem_uniform_key}. The cache hit probability and cache update algorithm were not addressed in \cite{DBLP:conf/iccnc/KiskaniS17}. 
The rest of the paper is organized as follows. In section \ref{related_sec}, the related works are described. In section \ref{netmod} the network model,  proposed encoding and  decoding  algorithms and earlier results on coded caching are presented. Section \ref{capacity_sec} focuses on the network capacity and the security aspect of coded caching is studied in section \ref{security_sec}. Cache hit probability for both coded and uncoded caching approaches is studied in section \ref{cache_hit}. In section \ref{cache_update}, an efficient caching update algorithm is proposed. 
Simulation results are provided in section \ref{sim_sec}. 
The paper is concluded in section \ref{conc_sec}. 

\section{Related Work}
\label{related_sec}

Many researchers have investigated the problem of caching in recent years. The fundamental limits of caching over a shared link is studied in \cite{DBLP:journals/tit/Maddah-AliN14}. The authors in \cite{DBLP:journals/tit/Maddah-AliN14} suggested to store uncoded contents or
uncoded segments of the contents in the caches during cache placement phase. During content delivery phase, they proposed to broadcast coded files to the users which resulted in  significant  multicasting gain. Such a gain is achievable by taking advantage of content overlap at  various caches in the network, created by a central coordinating server. This work was later extended in  \cite{DBLP:journals/ton/Maddah-AliN15} to scenarios with decentralized uncoded cache placement and there are many other related works that followed the same set of assumptions. 
While references  \cite{DBLP:journals/tit/Maddah-AliN14,DBLP:journals/ton/Maddah-AliN15} study the problem of caching in broadcast channels,  we  study the problem of caching in wireless ad hoc networks using multihop communications.

The authors in \cite{DBLP:conf/icc/JeonHJC15} studied the problem of caching in multihop networks. They assumed that during the cache placement phase, uncoded contents are stored in the caches independently in a decentralized fashion. They studied the throughput capacity of wireless ad hoc networks. In this paper, we  introduce decentralized coded caching for wireless ad hoc networks and compare our results with uncoded caching results in \cite{DBLP:conf/icc/JeonHJC15}. 

In decentralized uncoded caching scheme that we have considered in this paper, we assumed that different contents are selected uniformly at random and placed in the caches during the cache placement phase. We can therefore model the uncoded caching strategy as a coupon collector problem with group drawings as studied in \cite{stadje1990collector} and \cite{johnson1977urn}. Similar to the classical coupon collector problem, it is proved in \cite{johnson1977urn} that uniform selection of cached contents results in the minimum possible value for the average number of required hops. Hence, in this paper we study the case of uniform cache placement which is the best possible scenario in terms of network capacity.

Our proposed decentralized coded caching scheme utilizes random uniform LT codes \cite{DBLP:conf/focs/Luby02, mackay2005fountain, DBLP:journals/tit/Shokrollahi06} to encode the contents and store them during cache placement phase.

In random binary uniform LT codes \cite{mackay2005fountain}, a random code of length $m$ is chosen from $\mathbb{F}_2^m$ such that any of its  elements is either equal to zero or one. Such a random code can be represented by a random vector of length $m$ in $\mathbb{F}_2^m$. Then contents with indices corresponding to one in the random vector are added together and the encoded files are cached in the nodes.

Physical layer security has attracted many researchers in recent years. A survey of recent progress in this field can be found in \cite{bloch2011physical,DBLP:journals/wc/ShiuCWHC11}. Security aspects of network coding is studied in references like \cite{lima2007random,cai2002secure, vahidian2015relay}. 
In this paper, we will study the security of  LT coding-based caching technique. We specifically investigate the last hop communication security which is the most vulnerable transmission link. We prove that for a certain regime of caching sizes in the network, asymptotic perfect secrecy is achievable for the last hop\footnote{Security investigation for other hops remains as future work.}. Further, a secure caching update algorithm is proposed.   

Our proposed decentralized caching scheme can be used to create a distributed network coding based storage system. The concept of network coding for distributed storage was originally studied in \cite{DBLP:journals/tit/DimakisGWWR10,DBLP:journals/pieee/DimakisRWS11} where files are divided into packets and nodes need to collect all the packets to be able to reconstruct their desired contents. In our proposed distributed caching scheme, the requesting node needs to decode only one content. The authors in \cite{DBLP:journals/tit/DimakisGWWR10, DBLP:journals/pieee/DimakisRWS11} study the repair problem which is the problem of recovering data when some nodes fail. Further, while they address the repair problem using  MDS codes,  our proposed coded caching technique is based on {LT codes \cite{DBLP:conf/focs/Luby02}.}  This paper focuses on the scaling capacity and   security of a network with decentralized coded caching which is not studied in  \cite{DBLP:journals/tit/DimakisGWWR10,DBLP:journals/pieee/DimakisRWS11}. 

LT coding based storage is studied in references like \cite{DBLP:conf/infocom/CaoYYLH12,wang2012lt}. Reference \cite{DBLP:conf/infocom/CaoYYLH12} investigates the repair problem of LT codes in cloud services and \cite{wang2012lt} proposes new types of LT codes for storage systems. Use of fountain codes and Raptor codes \cite{DBLP:journals/tit/Shokrollahi06} has also been studied in  \cite{DBLP:conf/icassp/DimakisPR06,DBLP:journals/jsac/KongAS10}. None of these references have studied capacity, cache hit probabilities, cache update algorithms and security of LT code based storage in wireless ad hoc networks.

\section{Preliminaries}
\label{netmod}
\subsection{Network Model}
We consider a dense wireless ad hoc network in which $n$ nodes are uniformly distributed over a square of unit area as depicted in Figure \ref{fig_model}. These nodes use multihop communications to request one of $m$ contents from a set labeled as  $\mathcal{F}=\{F_1,F_2,\dots,F_m\}$. The requested content is denoted by $F_r$ and the minimum  number of nodes to decode $F_r$ by $N_r$. The minimum number of nodes to decode any requested content is denoted by $N$. We also assume that each node can cache $M$ files of equal size each containing $Q$ bits. In practice, the contents are divided into equal-sized chunks and the chunks are cached. Such an assumption is common in many references including \cite{DBLP:journals/tit/Maddah-AliN14, DBLP:journals/tvt/KiskaniS17, DBLP:journals/ton/Maddah-AliN15}.



\input{fig_model2}

The unit square area in Figure \ref{fig_model} is divided into many square-lets each with a side length of $c_1 s(n)$ where $s(n) = \sqrt{{\log(n)}/{n}}$. It is shown  \cite{DBLP:journals/tit/KulkarniV04} that each square-let contains  $\Theta(\log(n))$ nodes with probability 1. To avoid multiple access interference, a  { Protocol Model} \cite{DBLP:journals/ftnet/XueK06} is considered 
for  successful communication  between nodes. A Time Division Multiple Access (TDMA) scheme is assumed for the transmission between the square-lets. With the assumption of Protocol Model, if each square-let has a side length of $c_1 s(n)$ for a  constant $c_1$, then the square-lets with a distance of $c_2=\frac{2+\Delta}{c_1}$ square-lets apart can transmit simultaneously without significant interference \cite{DBLP:journals/tit/KulkarniV04} for a constant value $\Delta$.

For our proactive routing analysis, we assume that one of the nodes in each square-let is randomly chosen which is called { anchor node}. This node  collects all the information in the  square-let and combines them and relays the coded information to the next hop toward the requesting node. It is known that a minimum transmission range of $\Theta \left(\sqrt{\frac{\log(n)}{n}} \right)$ ensures network connectivity \cite{DBLP:journals/twc/KiskaniAS16} in such a dense network.

For the case of uncoded caching with proactive routing, \cite{DBLP:conf/icc/JeonHJC15} proposes a solution in which groups of neighboring square-lets will cooperate together to form a {\em local group}. It is proved  \cite{DBLP:conf/icc/JeonHJC15} that for a square local group with side length $s_g(n) = \Theta \left(\sqrt{\frac{m}{nM}}\right)$, any requested content is available in at least one node in the local group. A routing protocol \cite{DBLP:journals/tit/KulkarniV04} within the local group connects the source to destination through a series of horizontal and vertical square-lets. The paper assumes that all the nodes inside a local group have a global knowledge of cached contents within each local group. Such an assumption results in significant overhead which requires allocating part of network resources to the exchange of this information. Further, exchange of this information poses significant security threat and allows an eavesdropper to find out which contents are cached in the local group. 

In this paper, we propose a decentralized coded caching approach based on random vectors. In this technique, the contents are randomly and independently combined and stored in  caches. 
When a node requests a content, a unique linear combination of coded files can reconstruct the requested content. Therefore, each node first combines its encoded files and then transmits it to the anchor node (yellow circles in  Figure \ref{fig_model}). 
The anchor node also adds some information from its cache and forwards the newly updated file to the next anchor node closer to the requesting node
as shown in the lower left local group in Figure \ref{fig_model}. 
This process continues until the requesting node receives all the required coded files for decoding as shown in Figure \ref{fig_model2}.
For reactive routing approach, a random direction in the network is selected. Using the cached information in this random direction, the desired content can be obtained as shown in the upper left corner of Figure \ref{fig_model}. Interestingly, we prove selecting a random direction is optimal in terms of number of hops traversed to retrieve the content. In both cases, perfect communication secrecy is achievable. 

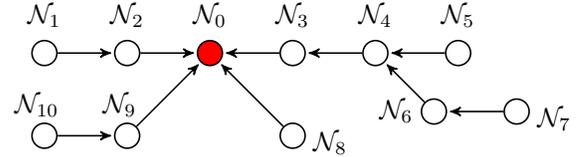
\begin{figure}[http]
 \centering
 \begin{tikzpicture}[->,>=stealth',shorten >=1pt,auto,node distance=1.1cm, semithick] 
  \node[circle,draw]  (N1) [] {};
  \node[circle,draw]  (N2) [right of=N1] {};
  \node[circle,draw,fill=red]  (N3) [right of=N2] {};
  \node[circle,draw]  (N4) [right of=N3] {};
  \node[circle,draw]  (Nl) [right of=N4] {};
  \node[circle,draw]  (Nl2) [right of=Nl] {};
  \node[circle,draw]  (No1) [below right of=Nl] {};
  \node[circle,draw]  (No2) [right of=No1] {};
   \node[circle,draw]  (No3) [below of=N4] {};
   \node[circle,draw]  (No4) [below of=N2] {};
   \node[circle,draw]  (No5) [left of=No4] {};

  \node (T3) [above of = N1, yshift=-0.6cm] {{$\mathcal{N}_{1}$}};
  \node (T4) [above of = N2, yshift=-0.6cm] {{$\mathcal{N}_{2}$}};
  \node (T5) [above of = Nl, yshift=-0.6cm] {{$\mathcal{N}_{4}$}};
  \node (T6) [above of = N3, yshift=-0.6cm] {{$\mathcal{N}_{0}$}};
  \node (T7) [above of = N4, yshift=-0.6cm] {{$\mathcal{N}_{3}$}};
  \node (T8) [above of = Nl2, yshift=-0.6cm] {{$\mathcal{N}_{5}$}};
  \node (T9) [left of = No1, xshift=0.6cm] {{$\mathcal{N}_{6}$}};
  \node (T10) [right of = No2, xshift=-0.6cm, yshift=-0.1cm] {{$\mathcal{N}_{7}$}};
  \node (T11) [right of = No3, xshift=-0.6cm, yshift=-0.1cm]
  {{$\mathcal{N}_{8}$}};
  \node (T11) [right of = No4, xshift=-1.2cm, yshift=0.4cm]
  {{$\mathcal{N}_{9}$}};
  \node (T11) [right of = No5, xshift=-1.2cm, yshift=0.4cm]
  {{$\mathcal{N}_{10}$}};

  \path (N1)   edge node {} (N2);
  \path (N2)   edge node {} (N3);
  \path (N4)  edge node {} (N3);
  \path (Nl)  edge node {} (N4);
  \path (Nl2)  edge node {} (Nl);
  \path (No1)  edge node {} (Nl);
  \path (No2)  edge node {} (No1);
  \path (No3)  edge node {} (N3);
  \path (No4)  edge node {} (N3);
  \path (No5)  edge node {} (No4);
\end{tikzpicture}
\caption{When a node $\mathcal{N}_{0}$ requests a file, it starts gathering the coded files from all the nodes in its local group to construct the requested content.} 
 \label{fig_model2}
 \vspace{-0.2in}
\end{figure}  

\subsection{Decentralized Coded Cache Placement}
In our proposed coded cache placement approach, a randomly encoded file $\mathbf{r}_j^i$ is created and placed in the $j^{th}$ cache location of node $i$. This randomly encoded file is a bit-wise summation of random contents from $\mathcal{F}$. In other words, 
\begin{equation}
   \mathbf{r}_j^i = \sum_{l=1}^{m} a_l^{i,j} F_l = {\bf v}_{j}^i \mathbf{F},
   \label{eq_def_1}
\end{equation}
 where $\mathbf{F} = [F_{1} ~F_{2} ~\dots~ F_{m}]^T$ represents the contents vector and ${\bf v}_{j}^i = [a_1^{i,j} ~ a_2^{i,j} ~ \dots ~ a_m^{i,j}]^T \in \mathbb{F}_2^{m}$ is a random row vector with each element equal to 0 or 1 and the summation is carried over Galois Field GF(2). This process is repeated independently for all cache locations of all nodes\footnote{In practice, each node chooses $M$ linearly independent random vectors during  cache placement phase. However, to simplify the analysis, we drop the independence assumption for different cache locations of a node. Therefore, analytical capacity results found in this paper are pessimistic.}. Notice that based on this construction, each vector $\mathbf{v}_{j}^i$ is uniformly selected from the set of all vectors in $\mathbb{F}_2^m$. Notice that this specific choice of fountain codes is known as Random Linear Fountain (RLF) codes \cite{mackay2005fountain} or Random Uniform LT codes \cite{DBLP:journals/tit/Shokrollahi06}.
  
 \subsection{Content Reconstruction} 
In order to decode any of the contents, nodes need to find $m$ linearly independent vectors ${\bf v}_{j}^i$ to span the m-dimensional space. For each requested content $F_r$, a unique combination of these encoded files will generate $F_r$. In both routing scenarios, the computation of appropriate gains for the combination of coded files is carried by the requesting node and this information is relayed to the neighboring nodes.
 
As depicted in Figure \ref{fig_model2}, node $\mathcal{N}_i$ in the routing path can contribute up to $M$ linearly independent row vectors ${\bf v}_1^i,{\bf v}_2^i,\dots,{\bf v}_M^i$ to span the entire space. Node $\mathcal{N}_i$  applies gain $b_j^i \in \{0,1\}$ to its $j^{th}$ cached file $\mathbf{r}_j^i$ and then adds (in binary field) $\sum_{j=1}^M b_j^i \mathbf{r}_j^i$ to the file it has received from previous hop and passes the newly constructed file to the next hop toward requesting node $\mathcal{N}_0$. After a total of $N_r$ transmissions, the  requesting node receives $ (\sum_{i=1}^{N_r} \sum_{j=1}^M b^i_j {\bf v}^i_j) \mathbf{F}$. It will then applies decoding coefficients to its own cached files and adds it to the received file to reconstruct the desired content. Note that each relay node adds some encoded files to the received file and relays it forward. The coefficients $b_j^i \in \{0,1\}$ are selected such that the linear combination of encoded files produce the desired requested content.  

\subsection{Prior Results}
\label{prior_sec}
The coded caching was originally introduced in \cite{kiskani2016capacity,DBLP:journals/twc/KiskaniS17} for cellular networks. The following lemma was proved in \cite{kiskani2016capacity,DBLP:journals/twc/KiskaniS17}.
 \begin{lem}
 {\em
 Let vector ${\bf v}^i_j \in \mathbb{F}_2^{m}$ {have} equiprobable elements. The average number of vectors ${\bf v}^i_j$ to span the $m$-dimensional space equals to
 \begin{equation}
    \mathbb{E}_{\textrm{uniform}} = m + \sum_{i=1}^{m} \frac{1}{2^{i}-1}  =  m + c_3,
    \label{eq_lem_res1}
 \end{equation}
 where $c_3$ asymptotically approaches the Erdős–Borwein constant 
 ($\approx 1.6067$).
 }\label{lem_uniform}
\end{lem}
Lemma \ref{lem_uniform} shows that the average number of cached files to reconstruct any content is equal to $m+c_3$ which is very close to $m$ for large values of $m$. This shows that random uniform vectors  perform close to optimal in terms of minimizing the number of cache locations to retrieve contents. Based on Lemma \ref{lem_uniform}, we have the following corollary.
\begin{corol}{\em 
 If random uniform vectors are used to create encoded files and then these files are independently cached in node caches, then on average 
 $\mathbb{E}[N]=(m+c_3)/M = \Theta({m}/{M})$ nodes are required to decode any requested content.
 }\label{corol_uniform}
\end{corol}
\begin{thm}{ \em
 In the proposed coded caching scheme, selecting a random direction is asymptotically optimal in terms of minimizing the number of hops required to retrieve all the contents.
} \label{thm_random_optimal}
\end{thm}
\begin{proof}
 In the proposed coded caching scheme, random vectors in $\mathbb{F}_2^m$ are used for encoding the contents. Lemma \ref{lem_uniform} shows that on average $m + c_3$ random vectors are needed to span the $m$-dimensional space of $\mathbb{F}_2^m$. Therefore, with the proposed decentralized coded caching scheme, on average  $\mathbb{E}[N]=\lceil (m+c_3)/M \rceil $ nodes are required to decode any requested content regardless of which routing direction is chosen for content retrieval. 
 
 On the other hand, to be able to retrieve all the contents, at least $m$ cache locations are necessary. This means that at least $\lceil m/M \rceil $ nodes are required for content retrieval in any caching scheme. Since $m+c_3$ is very close to $m$ for large values of $m$, then $ \Theta(\lceil (m+c_3)/M \rceil ) = \Theta (\lceil m/M \rceil) = \Theta(m/M)$, this proves that selecting any random direction is asymptotically optimal in terms of the minimum required number of hops for content retrieval. 
\end{proof}

Notice that Theorem \ref{thm_random_optimal} is intuitively valid since we have used completely random vectors and this means that all of the contents are equally distributed in all directions.

\section{Capacity}
\label{capacity_sec}
This section is dedicated to computation of network throughput capacity of decentralized coded and uncoded caching schemes for proactive and reactive routing techniques. First, we define the precise notions of achievable throughput and network throughput capacity as follows. 


\begin{mydef}{\em 
 A network throughput of $\lambda(n)$ contents per second for each node is {\em achievable} if there is a scheme for scheduling transmissions in the network, such that every content request by each node can be served by a rate of $\lambda(n)$ contents per second.
 }\label{def_achievable}
\end{mydef}

Whether a particular network throughput is achievable depends on the specific cache placement and node locations in the network. Since the location of the nodes and the cache placement in nodes is random, we will define the network capacity as the maximum asymptotic network throughput achievable with probability 1.

\begin{mydef}{\em 
 We say that the {\em throughput capacity} of the network is lower bounded by $\Omega(g(n))$ contents per second if a deterministic constant $c_5 > 0$ exists such that 
 \begin{align}
   \lim_{n \to \infty} ~\mathbb{P}[\lambda(n) = c_5 g(n) ~\textrm{is achievable}~] &= 1. 
 \label{lb_def_cap}
 \end{align}
 We say thet the network throughput capacity is upper bounded by $\operatorname{O}(g(n))$ contents per second if a deterministic constant  $c_6 < + \infty$ exists such that 
 \begin{align}
 \liminf_{n \to \infty} ~\mathbb{P}[\lambda(n) = c_6 g(n) ~\textrm{is achievable}~] &<1.
 \label{ub_def_cap}
 \end{align}
 We say that the network throughput capacity is of order $\Theta(g(n))$ contents per second if it is lower bounded by $\Omega(g(n))$ and upper bounded by $\operatorname{O}(g(n))$.
 }\label{def_cap}
\end{mydef}

In this paper, we study the throughput capacity after the cache placement phase  and during the content delivery phase. 
In the following, we will describe the necessary size of the local group to decode all the contents in proactive routing approach.
\begin{rem}{\em
Corollary \ref{corol_uniform} suggests that for decentralized coded caching on average $\Theta(m/M)$ nodes are needed to decode any desired content. Since any local group in proactive routing on average has $\Theta (n s_g(n)^2)$ nodes, then the average local group side length of $s_g(n) = \Theta \left(\sqrt{\frac{m}{nM}}\right)$ will be enough to decode all the contents. Notice that similar local group side length is found in \cite{DBLP:conf/icc/JeonHJC15} for the case of uncoded caching.
 }\label{rem_local_group_size}
\end{rem}
%
\subsection{Capacity of proactive routing approach}

In this section we assume that any node in each local group is completely aware of all the files cached in its local group. Nodes in local groups are cooperating with each other to transfer a requested content. 
In uncoded caching scenario it has been proved \cite{DBLP:conf/icc/JeonHJC15} that if $M \le m < nM$, a capacity of $\Theta \left( \sqrt{{M}/{m}} \right)$ is achievable. Our proposed decentralized coded caching approach achieves a capacity of $\Theta \left( {{M}/{(m\log(n))}} \right)$ while providing perfect secrecy as will be proved subsequently. 

\begin{thm}{\em
 In decentralized coded caching with a proactive routing scheme the following network throughput capacity is achievable
 \begin{equation}
  \lambda(n) = \Theta \left( \frac{M}{m\log(n)} \right).
  \label{eq_thm}
 \end{equation}
 \label{thm_capacity}
}\end{thm}
\begin{proof}
Corollary \ref{corol_uniform} shows that on average $\mathbb{E}[N]=\Theta({m}/{M})$ nodes are required to reconstruct all the contents in decentralized coded caching. Remark  \ref{rem_local_group_size} shows that each local group of side length $\Theta \left( \sqrt{\frac{m}{n M}}\right)$ contains this many nodes. As shown in Figure \ref{fig_model}, all these nodes cooperate to transmit the content to the requesting node. Since there are $\Theta(\log(n))$ nodes in each square-let, 
the total number of transmissions for one request in a local group is equal to $\Theta \left(\log(n) \left(\frac{s_g(n)}{s(n)}\right)^2\right)$. Therefore, the total number of file transmissions to satisfy all content requests in each local group has the order of $\Theta \left(\log(n) \left(\frac{s_g(n)}{s(n)}
\right)^2 \frac{m}{M} \right)$. On the other hand, in each local group we can have $\Theta\left(\frac{s_g(n)}{s(n)}\right)^2$ simultaneous transmissions. 
This implies that a network throughput of 
\begin{align}
 \lambda(n) = \Theta \left(\frac{\left(\frac{s_g(n)}{s(n)}\right)^2}{\log(n)
 \left(\frac{s_g(n)}{s(n)}\right)^2 \frac{m}{M}} \right) = 
 \Theta \left( \frac{M}{m\log(n)} \right)
\end{align}
is achievable.
\end{proof}

\subsection{Capacity of reactive routing approach}
Reactive routing usually requires less overhead but incurs higher delays in content delivery. However, one of the advantages of our proposed coded caching is that we can select any random direction and decode the desired content with the same optimum number of nodes. Such a scenario is depicted in the upper left corner of Figure \ref{fig_model}.   

For simplicity of our capacity analysis, we assume that all contents are of equal size with each having $Q$ bits. 
Assume that $\lambda(n)$ is the maximum achievable network throughput. This implies that with a probability close to one, the network can deliver $n \lambda(n)$ contents per second. Therefore, with a probability close to one, network nodes can transmit $n \lambda (n) \mathbb{E}[N] Q$ bits per second. There are exactly $\frac{1}{(c_2 c_1 s(n))^2}$ square-lets at any time slot available for transmission. Hence, the total number of bits that the network is capable of delivering is equal to $\frac{W}{(c_2 c_1 s(n))^2}$ where $W$ is the total available bandwidth. Therefore, 
\begin{equation}
 \lambda(n) = \dfrac{W}{n  \mathbb{E}[N] Q (c_2 c_1 s(n))^2} = \Theta 
 \left(\dfrac{1}{\mathbb{E}[N] \log n} \right),
 \label{eqs_cap_find}
\end{equation}
where $W$ and $Q$ are the total available bandwidth and total number of bits for each content respectively. 
This suggests that to find the maximum achievable network throughput, it is enough to find the average number of transmission hops needed to deliver the contents.


\begin{thm}{\em 
In decentralized coded caching with reactive routing, the network capacity is
 \begin{equation}
  \lambda(n) =   \Theta \left( \frac{M}{m \log(n)} \right).
 \label{cap_no_knowledge}
 \end{equation}
 }\label{thm_coded}
\end{thm}
\begin{proof}
Using Lemma \ref{lem_uniform}, on 
average $\mathbb{E}[N] = \Theta(\frac{m}{M})$ nodes (or equivalently hops) are required to decode a content in decentralized coded caching. This along with equation \eqref{eqs_cap_find} proves the theorem.
\end{proof}
In decentralized uncoded caching strategy, nodes cache contents with uniform distribution. Lemma \ref{leme0} computes the average number of hops to retrieve a content.
\begin{lem}
 {\em  If a content is requested independently and uniformly at random from a set of $m$ contents, then the average  number of requests to have at least one copy of each content is equal to $ \mathbb{E}[l] =  m \sum_{i=1}^{m} \frac{1}{i} = m H_{m} = \Theta(m \log(m))$. 
 }
 \label{leme0}
\end{lem}
\begin{proof}
\textcolor{black}{This is the well-known {\em coupon collector problem} \cite{erdHos1961classical}.}
\end{proof}
\begin{lem}{\em
 If each node stores $M$ different files uniformly at random during cache placement phase, then the average number of nodes required so that each file is cached in at least one node is between $\dfrac{m H_m}{d(m,M)}$ and $1+\dfrac{m H_m}{d(m,M)}$ where 
  \begin{equation}
  d(m,M) = \sum_{j=0}^{M-1} \frac{m}{m-j}.
  \label{eq_lem_batch}
 \end{equation}
 }\label{lem_batch}
\end{lem}
\begin{proof}
 This problem is an extension of the coupon collector problem in Lemma \ref{leme0} as the files in each node's cache are different. To find the average number of nodes to have one copy of each file in at least one node, we start from a classic coupon collector problem. Assume that files are chosen uniformly at random and as soon as $M$ different files are found, they are placed in a node's cache. Then the same process is  started over for the next node and after finding $M$ distinct files, the files are placed in its cache. Assume that this process is repeated until  one copy of each file is cached in at least one node. This is a geometric distribution problem, then on average $d(m,M)$ files are required to fill up one node's cache. Based on Lemma \ref{leme0},   we will have a copy of each file in at least one node's cache after an average $m H_m$ file requests. Hence, the average number of nodes required to have one copy of each file in at least one node is between $\dfrac{m H_m}{d(m,M)}$ and $1+\dfrac{m H_m}{d(m,M)}$. 
\end{proof}

\begin{thm}{\em 
 If $m>>M$, the capacity of the network using decentralized uncoded caching is equal to  
  \begin{equation}
  \lambda = \Theta \left(\frac{M}{m \log(m) \log(n)} \right).
  \label{ex_uncoded}
 \end{equation}
  \label{thm_uncoded}
}\end{thm}
\begin{proof}
 Lemma \ref{lem_batch} shows that in case of uncoded caching the average number of nodes needed to satisfy all requests 
 is upper and lower bounded as 
 \begin{align}
\dfrac{m H_m}{d(m,M)} \le \mathbb{E}[N] \le 1+\dfrac{m H_m}{d(m,M)}.  
 \end{align}
 Therefore, for large values of $m$, the average number of nodes  for decoding scales as 
 \begin{align}
 \mathbb{E}[N] =  \Theta \left(\dfrac{m H_m}{d(m,M)} \right).
 \label{eqs_thm_1_2}
 \end{align}
 To find a bound for $d(m,M)$, notice that the series in the right hand side of equation \eqref{eq_lem_batch} has $M$ terms and the maximum term $\frac{m}{m-M+1}$ corresponds to the case when $j=M-1$ and the minimum term 1 corresponds to the case when $j=0$. A lower bound and an upper bound on $d(m,M)$ can be found by using the minimum and maximum terms respectively. Therefore,
 \begin{equation}
  M \le d(m,M) \le \frac{M m}{m-M+1}.
  \label{eq_lab}
 \end{equation}
 When $m>>M$, then the lower and upper bounds of equation \eqref{eq_lab} converge to the same value of $M$.
  \begin{align}
  \mathbb{E}[N] &= \Theta \left(\frac{m \log(m) }{M } \right) 
  \label{eqs_lb_thm1_2}
  \end{align}
 Combining  \eqref{eqs_lb_thm1_2} and \eqref{eqs_cap_find} proves the theorem. 
 \end{proof}
\begin{rem}{\em
Theorems \ref{thm_coded} and \ref{thm_uncoded} show that  decentralized coded caching strategy can improve the network capacity by a factor of $\log(m)$ in reactive routing.  
}\label{rem_cap_improve}
\end{rem}

Figure \ref{fig_theory} compares the capacity of coded caching with uncoded caching for both proactive and reactive routing algorithms. To plot this figure, we assumed that the number of contents $m$, grows polynomially with the number of nodes $n$. 
Coded caching provides perfect secrecy as will be discussed later while it performs better (worse) than uncoded caching for reactive (proactive) routing algorithm. 
\begin{figure}[http]
    \center
      \includegraphics[scale=0.5,angle=0]{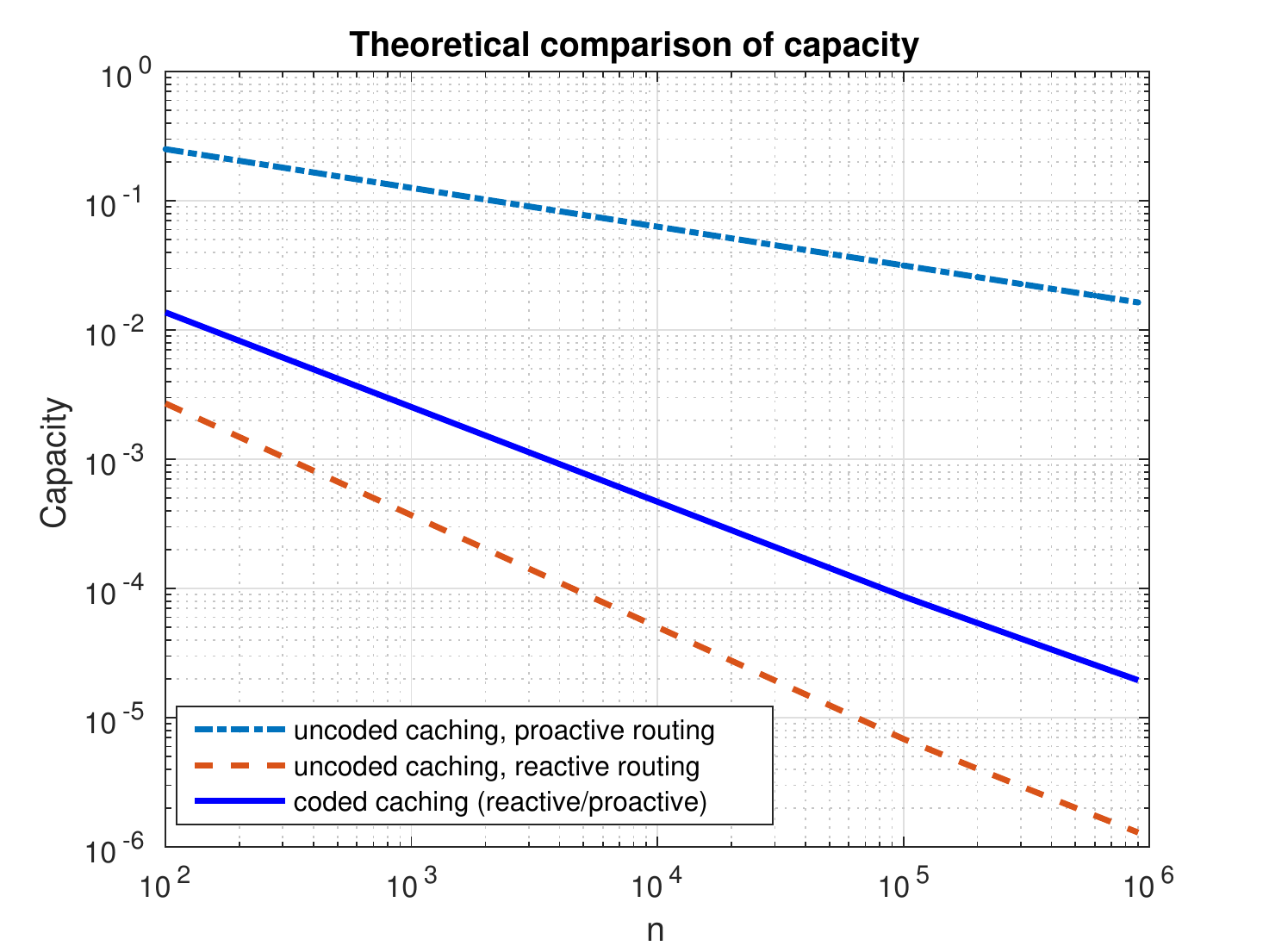}
\caption{Capacity for coded and uncoded caching using proactive and reactive routing algorithms.} 
\label{fig_theory}
\end{figure}

\section{Security} 
\label{security_sec}
 This section evaluates the security of  coded caching strategy. Note that uncoded caching allows an adversary with a noiseless wiretap channel to perfectly receive the transmitted files. We prove that such an eavesdropper will not be able to reduce its   equivocation about the transmitted files in coded caching approach when there is a large number of files. Therefore, asymptotic perfect secrecy can be achieved. This problem was originally studied by Shannon \cite{shannon1949communication}.
 
Our secrecy proof is applicable to both proactive and reactive routing schemes. As described earlier, each node combines  its encoded files as $ \mathbf{x}^i = \sum_{j=1}^M b_j^i \mathbf{r}_j^i = \sum_{j=1}^M b_j^i \mathbf{v}_j^i \mathbf{F} $ and adds it to the previously received file and forwards it to the next hop. Therefore, the aggregate received file by the requesting node $\mathcal{N}_0$ is 
$S_r = \sum_{i=1}^{N_r} \mathbf{x}^i = \sum_{i=1}^{N_r} \sum_{j=1}^M b_j^i \mathbf{r}_j^i.$
We assume that the encoding vectors $\mathbf{v}_j^i$ of the neighboring nodes are transmitted to the requesting node $\mathcal{N}_0$ through a secure low bandwidth channel. When enough number of such vectors are gathered, $\mathcal{N}_0$ computes the decoding coefficients $b_j^i$ in order to generate the desired file $F_r$. Then, it sends back the $b_j^i$ coefficients through the low bandwidth secure channel to its neighboring nodes. The secure channel is used only to transmit the encoding and decoding information. Transmitting the large encoded files through secure channel would be undesirable due to low bandwidth constraint. However, the encoding vectors $\mathbf{v}_j^i$ and the decoding gains $b_j^i$ have much smaller sizes for transmission through secure channel.

If content $F_r$ is used in the encoding of at least one of the coded cached files in $\mathcal{N}_0$ (i.e. $a_r^{0,j}=1$ for some $1 \le j \le M$), then the requesting node $\mathcal{N}_0$ can generate the file 
\begin{align}
 \mathbf{x}^0_r = \sum_{j=1}^M b_j^0 \mathbf{r}_j^0 = \sum_{j=1}^M b_j^0 \mathbf{v}_j^0 \mathbf{F} = F_r + \sum_{j=1}^M b_j^{'0} \mathbf{v}_j^0 \mathbf{F}
\end{align}
from its own encoded cached files. Note that $\mathbf{v}_{\textrm{req}}=\sum_{j=1}^M b_j^{'0} \mathbf{v}_j^0$ is a coding vector in $m$-dimensional space. Node $\mathcal{N}_0$ only needs to receive $S_r = \mathbf{v}_{\textrm{req}}  \mathbf{F}$ and add  $S_r$ to $\mathbf{x}^0_r$ in GF(2) to retrieve $F_r$. The requesting node $\mathcal{N}_0$ uses the secure channel to collect enough number of encoding vectors in order to span $\mathbf{v}_{\textrm{req}}$. Then the requesting node $\mathcal{N}_0$ finds the appropriate decoding coefficients $b_j^i$ to span $\mathbf{v}_{\textrm{req}}$ and  sends these decoding coefficients back to the neighboring nodes through the secure channel. The neighboring nodes collaboratively create the right decoding file $S_r$ and transmit it to $\mathcal{N}_0$. 

The second possibility which is less likely to happen for large values of $M$ is that none of the encoded files in $\mathcal{N}_0$ contains $F_r$ (i.e. $a_r^{0,j}=0$ for all $1 \le j \le M$). In that case, $\mathcal{N}_0$ generates a unique combination of its encoded files as $\mathbf{x}^0_r=\sum_{j=1}^M b_j^0 \mathbf{v}_j^0 \mathbf{F}$. In order to decode $F_r$, node $\mathcal{N}_0$ needs to receive $S_r = F_r+\mathbf{x}_r^0$. Hence, it uses the secure channel to collect enough number of encoding vectors $\mathbf{v}_j^i$ to be able to construct $S_r$.
After solving the linear equation in GF(2), it  sends the decoding gains back to the neighboring nodes such that they can collaboratively create the encoded file $S_r = F_r+\mathbf{x}_r^0$ which allows $\mathcal{N}_0$ to retrieve $F_r$.

We claim asymptotic perfect secrecy for this approach is achievable as long as for each requested file $F_r$, the requesting node generates a different encoded combination $\mathbf{x}^0_r$ which acts similar to key and hence, the transmitted signal is in fact an encrypted version of the message $F_r$. The intended receiver is indeed capable of decrypting the message by adding its secret key $\mathbf{x}^0_r$ to it. Notice that for each requested file $F_r$, a different key $\mathbf{x}^0_r$ is generated using the encoded cached files in the requesting nodes. No other eavesdropper will be able to decode the message as they do not have the secret key $\mathbf{x}^0_r$. {The main advantage of the legitimate receiver is the information stored in its cache which allows it to create a unique key $\mathbf{x}^0_r$ for each requested file $F_r$. This is similar to Shannon cipher problem \cite{shannon1949communication}. 

In \cite{shannon1949communication}, Shannon introduced the Shannon cipher system in which an encoding function $\mathfrak{e}: \mathbb{M} \times \mathbb{K} \to 
\mathbb{C}$ is mapping a message $\mathfrak{M} \in \mathbb{M}$ and a key $\mathfrak{K} \in \mathbb{K}$ to a codeword $\mathfrak{C} \in \mathbb{C}$. In our problem, for each requested {content} by a user, the legitimate receiver uses a unique key $\mathfrak{K}$ to recover the message $\mathfrak{M}$. Even when a different user requests the same file, it uses a different key because each user caches different encoded files. 
The unique key for each user depends on the coded files that the node is storing and the coded files from neighboring nodes used for decoding the requested file. Shannon proved that if a coding scheme for Shannon's cipher system achieves  perfect secrecy, then $\mathbb{H}(\mathfrak{K}) \ge \mathbb{H}(\mathfrak{M})$ where $\mathbb{H}(.)$ denotes the entropy. He proved that at least one secret key bit should be used for each message bit to achieve perfect secrecy. If the sizes of messages, keys and codewords are the same, there are necessary and sufficient conditions  \cite{bloch2011physical} to obtain perfect  secrecy presented in the  following theorem. 
\begin{thm}{\em 
 If $|\mathbb{M}|=|\mathbb{K}|=|\mathbb{C}|$, a coding scheme achieves perfect secrecy if and only if 
 \begin{itemize}
  \item For each pair $(\mathfrak{M}, \mathfrak{C}) 
  \in (\mathbb{M} \times \mathbb{C})$, there exists a unique key 
  $\mathfrak{K} \in \mathbb{K}$ such that $\mathfrak{C} = 
  \mathfrak{e}(\mathfrak{M},\mathfrak{K})$.
  \item The key $\mathfrak{K}$ is uniformly distributed in $\mathbb{K}$.
 \end{itemize}
 }\label{thm_shannon}
\end{thm}
\begin{proof}
 The proof can be found in section 3.1 of \cite{bloch2011physical}.  
\end{proof}
We will use Theorem \ref{thm_shannon} to prove that our approach can achieve asymptotic perfect secrecy. In either of the cases, the requesting node $\mathcal{N}_0$ receives a codeword\footnote{This is true for both scenarios because all the operations are in GF(2).} $S_r = F_r + \mathbf{x}_r^0$ from the last node adjacent to  $\mathcal{N}_0$. Node $\mathcal{N}_0$ uses XOR operation to decode $F_r$ from $S_r$ using it's secret key $\mathbf{x}^0_r$. Therefore, we have a Shannon cipher system in which $\mathfrak{M}=F_r, \mathfrak{K}=\mathbf{x}^0_r, \mathfrak{C} = S_r$ and $\mathfrak{e}$ denotes the XOR operation.  
To use this theorem, first we prove that for large enough values of $m$, the key $\mathbf{x}_r^0$ is uniformly distributed.
\begin{lem}{\em
The asymptotic distribution of bits of coded files in caches tend to uniform.
 }\label{lem_uniform_key}
\end{lem}
\begin{proof}
We assume that all files have $Q$ bits and they may have a distribution different from uniform. We will prove that each coded cache file will be uniformly distributed for large values of $m$. Let us denote the $k^{th}$ bit of file $F_l$ by $f_l^k$
 where $1 \le k \le Q$ and $1 \le l \le m$. Assume that 
$\mathrm{Pr}(f_l^k = 1) = p_l^k = 1 - \mathrm{Pr}(f_l^k = 0)$. 
Further, we assume that the bits of files ($f_l^k$) are independent. The $k^{th}$ bit of the coded file in the $j^{th}$ cache location of node $i$ can be represented as
\begin{equation}
  r^{k}_{i,j} = \sum_{l=1}^m a_l^{i,j} f_l^k,
  \label{coded_file2}
 \end{equation}
 where $a_l^{i,j}$ is a binary value with  uniform distribution and independent of all other bits. Using regular summation (not over GF(2)) and denoting 
 $h^{i,j,k}_l \triangleq a_l^{i,j} f_l^k$, we define $ H^{i,j,k} \triangleq \sum_{l=1}^m h^{i,j,k}_l$. Therefore, $
  \mathrm{Pr}[r^k_{i,j} = 0] = \mathrm{Pr} [H^{i,j,k} \stackrel{2}{\equiv} 0]$.
Therefore, the $k^{th}$ bit of the coded file is equal to 0 if an even number of terms in $H^{i,j,k}$ is equal to 1. The probability distribution of $H^{i,j,k}$ can be computed using 
 {\em probability generating functions}. Since $h^{i,j,k}_l$ is a
 Bernoulli random variable with probability $\frac{1}{2} p_l^k$, its probability generating function is equal to  
 \begin{equation}
  G_{l}^{i,j,k} (z) = (1-\frac{1}{2} p_l^k)  + \frac{1}{2} p_l^k z.
  \label{pgf_bernoulli}
 \end{equation}
 Since $a_l^{i,j}$ and $f_l^k$ are independent random variables, $h^{i,j,k}_l$ will become independent random variables. Therefore, the probability generating function of $H^{i,j,k}$ denoted by $G_H^{i,j,k}(z)$ is the product of all probability generating functions.
  \begin{equation}
  G_H^{i,j,k} (z) = \prod_{l=1}^m \left( (1-\frac{1}{2} p_l^k)  + 
  \frac{1}{2} p_l^k z \right).
  \label{pgf_rijkl} 
 \end{equation}
Denoting the probability distribution of $H^{i,j,k}$ as $\mathfrak{h}(.)$, the probability of $H^{i,j,k}$ being even is
 \begin{align}
 & \mathrm{Pr}[H^{i,j,k} \stackrel{2}{\equiv} 0] = 
  \sum_{u=0}^{\lfloor \frac{m}{2} \rfloor} \mathfrak{h}(2 u) =
  \sum_{u=0}^{\lfloor \frac{m}{2} \rfloor} \mathfrak{h}(2 u) z^{2 u}
  \biggr\rvert_{z=1}  \nonumber \\
  &= \frac{1}{2} \left[
  \sum_{u=0}^{m} \mathfrak{h}( u) z^{ u} +  \sum_{u=0}^{m} \mathfrak{h}( u) (-z)^{ u}
  \right]_{z=1}  \nonumber \\
&=\frac{1}{2} G_H^{i,j,k} (1) +  \frac{1}{2} G_H^{i,j,k} (-1) 
=\frac{1}{2}\prod_{l=1}^m \left( (1-\frac{1}{2} p_l^k)  + 
  \frac{1}{2} p_l^k  \right) \nonumber \\
 &+  \frac{1}{2} \prod_{l=1}^m \left( (1-\frac{1}{2} p_l^k) 
  - \frac{1}{2} p_l^k  \right) 
  =\frac{1}{2} \left( 1 + \prod_{l=1}^m  \left(1-p_l^k \right) \right) \nonumber 
\end{align}
Therefore,
\begin{align}
 &\lim_{m \to \infty} \mathrm{Pr}[r^k_{i,j} = 0] 
 = \lim_{m \to \infty} \frac{1}{2} \left( 1 + \prod_{l=1}^m 
 \left(1-p_l^k \right) \right) = \frac{1}{2}+ \nonumber \\
 & \frac{1}{2} \lim_{m \to \infty}  \prod_{l=1}^m 
 \left(1-p_l^k \right) = \frac{1}{2}+ \frac{1}{2} \lim_{m \to \infty}   \left(1-\inf \{p_l^k\} \right)^m = \frac{1}{2}. \nonumber
 \label{proof_end}
\end{align}
This proves the lemma.
\end{proof}
This lemma paves the way to prove the following theorem. 
\begin{thm}{\em The proposed coded caching strategy provides asymptotic perfect secrecy for the last  hop if $m$ is large and $m <2^M$.
}\label{thm_secrecy}
\end{thm}
\begin{proof}
To formulate this as a Shannon cipher problem, we assume that $\mathfrak{M}=F_r$, $\mathfrak{K}=\mathbf{x}^0_r$, and $\mathfrak{C} = S_r$. The condition $m < 2^M$ ensures that a unique key exists for each requested message since at most $2^M$ possible random keys can be built from  $M$ cached files. The encoding function is XOR operation. For any pair $(\mathfrak{m}, \mathfrak{C}) \in (\mathbb{M}, \mathbb{C})$, a unique key $\mathfrak{K} \in \mathbb{K}$ exists such that $\mathfrak{C} = \mathfrak{m} + \mathfrak{K}$ which guarantees that $|\mathbb{M}|=|\mathbb{K}|=|\mathbb{C}|$.
 
Notice that the key  $\mathfrak{K}=\mathbf{x}^0_r$ belongs to the set of all possible bit strings with $Q$ bits. Lemma \ref{lem_uniform_key} proves that each {coded cached content is uniformly distributed among all $Q$-bit strings. Hence each key which} is a unique summation of cached encoded files {is uniformly distributed among the set of all $Q$-bit strings.}  In other words, regardless of the distribution of the bits in files, $\mathbf{x}^0_r$ can be any bit string with equal probability for large values of $m$. Therefore, conditions of Theorem \ref{thm_shannon} are met and asymptotic perfect secrecy is achieved. 
\end{proof}


\begin{rem}{\em
 In this paper, we only studied the security of last hop communications in our approach and proved that even in the most vulnerable (last) link, secure communications is possible. A more general security study for all links remains as future work. Also, the study of the security of this approach against cooperative eavesdroppers remains as future work. 
}\label{rem_last_hop}
\end{rem}

\section{Cache Hit Probability}
\label{cache_hit}
This section is dedicated to computation of cache hit probability when a node $\mathcal{N}_0$ can access $u$ other nodes $\mathcal{N}_1,\dots,\mathcal{N}_u$ or equivalently, $l=u M$ cache locations. We compute the event that $\mathcal{N}_0$ can decode any desired file in the set $\mathcal{F}$ with this information. First, let's define the cache hit probability.  
\begin{mydef}{\em
 The \emph{cache hit probability for all contents} is defined as the probability that any content can be retrieved using the cached information in nodes $\mathcal{N}_1,\dots,\mathcal{N}_u$. 
 }\label{def_hit_all}
\end{mydef}
We will first study uncoded cache hit probability. 
\subsection{Uncoded Caching}
In uncoded caching, each node is randomly choosing $M$ different contents from the set of $m$ contents. This can be modeled as a {coupon collector problem with group drawings} in which a coupon collector is collecting a number of different coupons in each time and wants to find the probability that after $u$ group collections, all the contents are collected. The result studied before \cite{stadje1990collector, johnson1977urn} and is summarized below. 
\begin{thm}{\em
 Assume that a coupon collector collects $m$ different coupons. Each time a bundle of $M$ different coupons are drawn uniformly at random. If $u$ bundles are collected, then the probability that all the coupons are collected is equal to 
 \begin{equation}
  \mathbb{P}_{\textrm{all collected}} = \sum_{j=0}^{m-M}(-1)^j \binom{m}{j} \left( \frac{\binom{m-j}{M}}{\binom{m}{M}} \right)^u
  \label{eqs_group_coupon}
 \end{equation}

 }\label{thm_urn}
\end{thm}
\begin{proof}
 The proof can be found in \cite{stadje1990collector} (a special case of Theorem 1) and also in page 164 of \cite{johnson1977urn}.
\end{proof}
Result of Theorem \ref{thm_urn} is the cache hit probability for collecting all contents when uncoded caching is used. 
\subsection{Coded caching}
In coded caching, we compute the probability that there exists $m$ linearly independent vectors within $l = u M $ random  encoding vectors.
If $m > l = u M$, this probability is clearly zero. Therefore, without loss of generality, we compute this probability when  $l = u M \ge m$. 
This problem has been studied in  literature \cite{kolchin1999random} and the results are summarized below.
\begin{thm}{\em 
Let $l \ge m \ge 1$ and $s $ be positive integers and $r = l - m$. If $A=[a_{ij}]$ is an $l \times m$ matrix whose elements are independent binary uniform random variables and $\rho_m(l)$ is the rank of matrix $A$ in GF(2), then if $m \to \infty$ we have  
\begin{align}
 \mathbb{P}[\rho_m(l) = m - s] \to 2^{-s(s+r)} \prod_{i=s+1}^{\infty}  \left(1 - \frac{1}{2^i} \right) \nonumber \\ 
 \times \prod_{j=1}^{r+s} \left(1 - \frac{1}{2^j} \right)^{-1}, 
 \label{eqs_kolchin_eqsss1}
\end{align}
where the last product equals 1 for $s+r=0$.
} \label{thm_kolchin}
\end{thm}
\begin{proof}
This is Theorem 3.2.1 in page 126 of \cite{kolchin1999random}.
\end{proof}
\begin{corol}{\em 
Let $l \ge m$ and $A=[a_{ij}]$ be an $l \times m$ matrix whose elements are independent binary uniform random variables and $\rho_m(l)$ be the rank of matrix $A$ in GF(2). If $m \to \infty$, then 
 \begin{align}
 \mathbb{P}[\rho_m(l) = m] \to  \prod_{i=l-m+1}^{\infty} \left(1 - \frac{1}{2^i} \right). 
 \label{eqs_fullrank_corol}
\end{align}
 }\label{corol_kolchin_fullrank}
\end{corol}
Equation \eqref{eqs_fullrank_corol} is the cache hit probability for coded caching approach. 

\begin{rem}{\em 
The cache hit probability for coded caching very quickly approaches 1 if $l$ is slightly larger than $m$. In fact, there is a very sharp transitioning of the probability from 0 to 1 in coded caching around the point $l = m$. However, in uncoded caching, $l$ should be much larger than $m$ in order for the cache hit probability tends to 1 (see Figure \ref{fig_m100}). 
}\label{rem_compare_cache_hit}
\end{rem}

\begin{rem}{\em
This result demonstrates that coded caching scheme utilizes the cache space efficiently and avoids over-caching unlike uncoded caching approach. 
}\label{rem_overcaching} 
\end{rem}

\section{Cache Update Algorithm}
\label{cache_update}
In this section, a caching update algorithm is described. Let's assume a new content $F_{\textrm{new}}$ should replace another content $F_{k}$ based on some replacement policy {such as  Least Recently Used (LRU) or  Least Frequently Used (LFU) policy}. {The network  controller uses a bit scrambling technique to create a file $F_{\textrm{new}}^{'}$ with uniform bit distribution from $F_{\textrm{new}}$. Such bit scrambling techniques are widely used in communication systems to give the transmitted data useful engineering properties \cite{hui2003method}. 
Notice that the bit scrambling technique makes $F_{\textrm{new}}^{'}$ equivalent of temporary secret key with uniform distribution.} 
The network  controller then generates $F_{k} + F_{\textrm{new}}^{'}$ and broadcasts this file to the network nodes. When node $N_i$ receives $F_{k} + F_{\textrm{new}}^{'}$, it will add $F_{k} + F_{\textrm{new}}^{'}$ to all of its cached encoded files $\mathbf{r}^i_j$ which contain $F_{k}$, i.e., all $\mathbf{r}_j^i$'s for which $a_k^{i,j}=1$. In other words, if in the $j^{th}$ location of node $i$ we have 
\begin{equation}
   \mathbf{r}_j^i = F_k+\sum_{\substack{l=1 \\ l \neq k}}^{m} a_l^{i,j} F_l, 
   \label{eq_def_end}
\end{equation}
then $F_{k} + F_{\textrm{new}}^{'}$ will be added to $\mathbf{r}^i_j$. This replaces $F_k$ with $F_{\textrm{new}}^{'}$ in all encoded files that contains $F_k$. If some cached encoded files does not contain $F_k$, then no action is required for those encoded files. {Nodes can then decode $F_{\textrm{new}}^{'}$ using the same decoding gains as for $F_k$ without any additional overhead. When $F_{\textrm{new}}^{'}$ is decoded, then a de-scrambling algorithm can be used to recover $F_{\textrm{new}}$ from $F_{\textrm{new}}^{'}$.} 
A pseudocode representation of our caching update protocol is shown in Algorithm \ref{alg:1}.

\begin{algorithm}
    \caption{Cache Update Algorithm}
    \label{alg:1}
    \begin{algorithmic}[1]
        \Procedure{Cache Update}{}
        \State {\bf Find} {the content $F_{k}$ that should be replaced}. 
        \State {\bf  Scramble} {$F_{\textrm{new}}$ to get $F_{\textrm{new}}^{'}$ with uniform bits.}
        \State {\bf Encode} the new content $F_{\textrm{new}}^{'}$ with $F_{k}$ as $F_{\textrm{new}}^{'} \oplus F_{k}$.
        \State {\bf Broadcast} $F_{\textrm{new}}^{'} \oplus F_{k}$ to all the nodes.
        \For{ the $j^{th}$ cache location of node $i$} 
        \If{$F_{k}$ is used in encoding $\mathbf{r}^i_j$ (i.e. $a_k^{i,j}=1$)} 
        \State {\bf Update} the $j^{th}$ cache location of node $i$ with  
        \Statex~~~~~~~~~~~~ $F_{\textrm{new}}^{'} \oplus F_{k} \oplus \mathbf{r}^i_j$. 
        \EndIf
        \EndFor 
        \EndProcedure
	\Statex       
    \end{algorithmic}
\end{algorithm}

{ Notice that during the caching update phase none of the contents is transmitted and any eavesdropper would only receive encoded version of the files. Therefore, with this caching update technique, the contents can be updated securely.}

\section{Simulation}
\label{sim_sec}
This section verifies the analytical results derived earlier via simulations.
%
Figure \ref{fig_hop} compares the average number of hops required to decode the contents in decentralized coded and uncoded caching schemes.
We consider a wireless ad hoc network with $n=1000$ nodes and $m=100$ contents. 
The simulation results clearly demonstrate that decentralized coded caching outperforms uncoded case particularly when the cache size is small which is the most likely operating regime. For instance, with decentralized coded content caching, a cache of size 25 only requires less than 5 hops while decentralized uncoded content caching needs around 20 hops for successful content retreival. This makes coded content caching much more practical compared to uncoded content caching. Note that capacity is inversely proportional to the average hop counts. As can be seen from Figure \ref{fig_hop}, for small cache sizes, coded caching significantly reduces the number of hops required to decode the contents. This property is important for nodes with small storage capability since large number of hops can impose excessive delay and low quality of service.

Figure \ref{fig_hop} proves another important result that content retrieval can be optimally done in any random direction. In this plot, we have used random directions of east, west, south and north for content retrieval using coded caching and shown that the average number of hops in any of these random directions is the same. The four plots corresponding to these four directions is so close that it is hard to distinguish  them.

\begin{figure}
    \center
      \includegraphics[scale=0.5,angle=0]{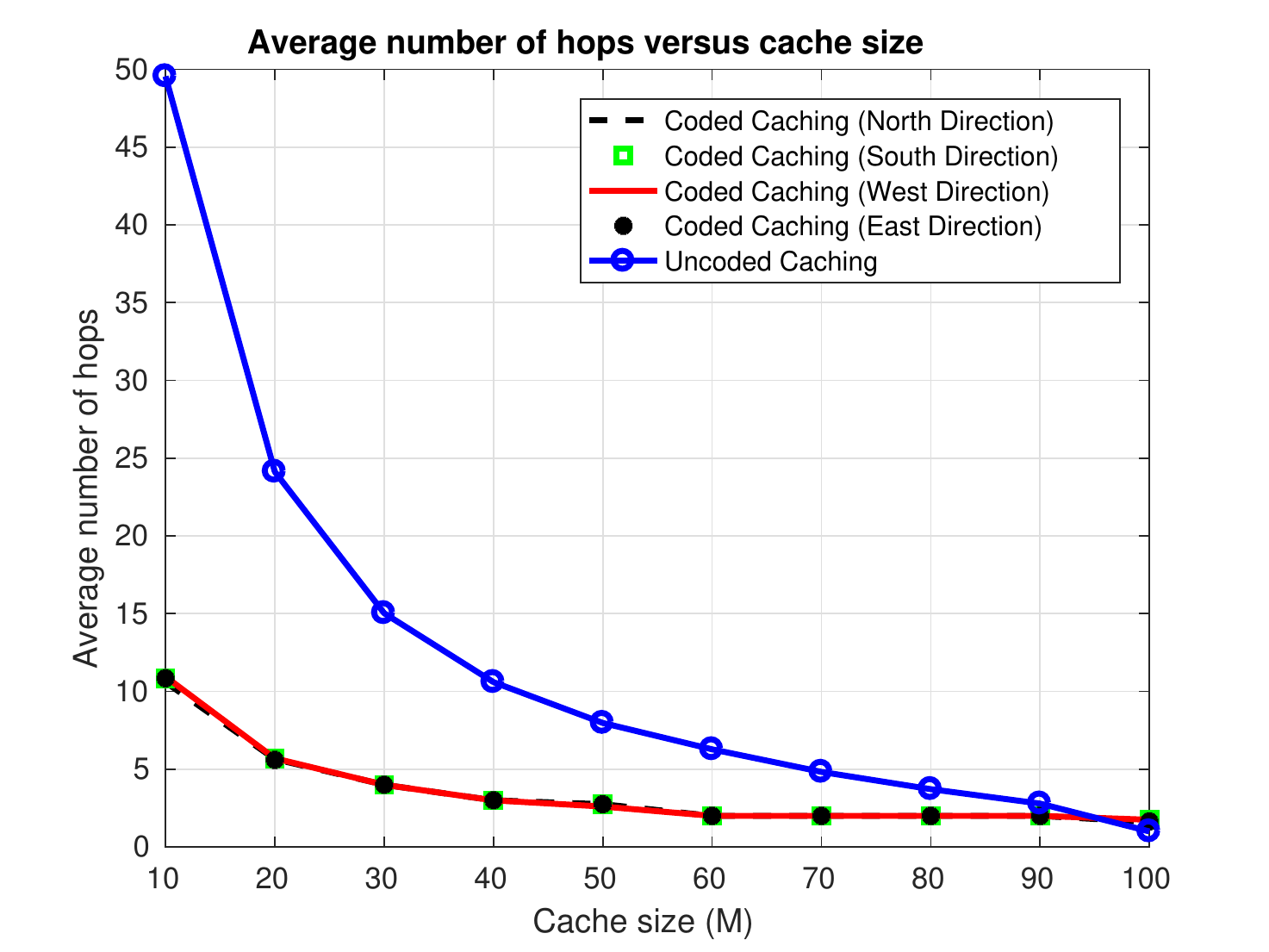}
\caption{Average required number of hops (nodes) for coded and uncoded caching schemes to retrieve a content.}
\vspace{-0.2in}
\label{fig_hop}
\end{figure}

 Figure \ref{fig_m100} shows the simulation results for different values of $M$ when $m=100$. \textcolor{black}{The cache hit probability is plotted as a function of the number of  cached contents $l$ and for different cache sizes $M$. However, for each fixed value of $l$, the number of nodes  $\mu = \frac{l}{M}$ will be different depending on $M$. For instance, at $l=400$, $M$ and $\mu$ are 25 and 16 respectively.} The simulation results are validating the theoretical results in {Theorem}  \ref{thm_urn} and {Corollary}  \ref{corol_kolchin_fullrank}. As can be seen from {this figure}, the cache hit probability for coded caching is much higher than that of uncoded caching. Also, it is clear from this plot that the cache hit probability of coded caching  approaches 1  rapidly when $l$ starts to be greater than $m$. For uncoded caching, specially when $m$ is large, the receiver should access a much larger number of cached contents in order for the cache hit probability to approach 1. Therefore, in networks with a large number of contents our coded caching approach will quickly achieve a cache hit probability close to one with a much smaller number of cache locations. This is a significant benefit of our technique in reducing  overcaching.  

\begin{figure}
    \center
      \includegraphics[scale=0.45,angle=0]{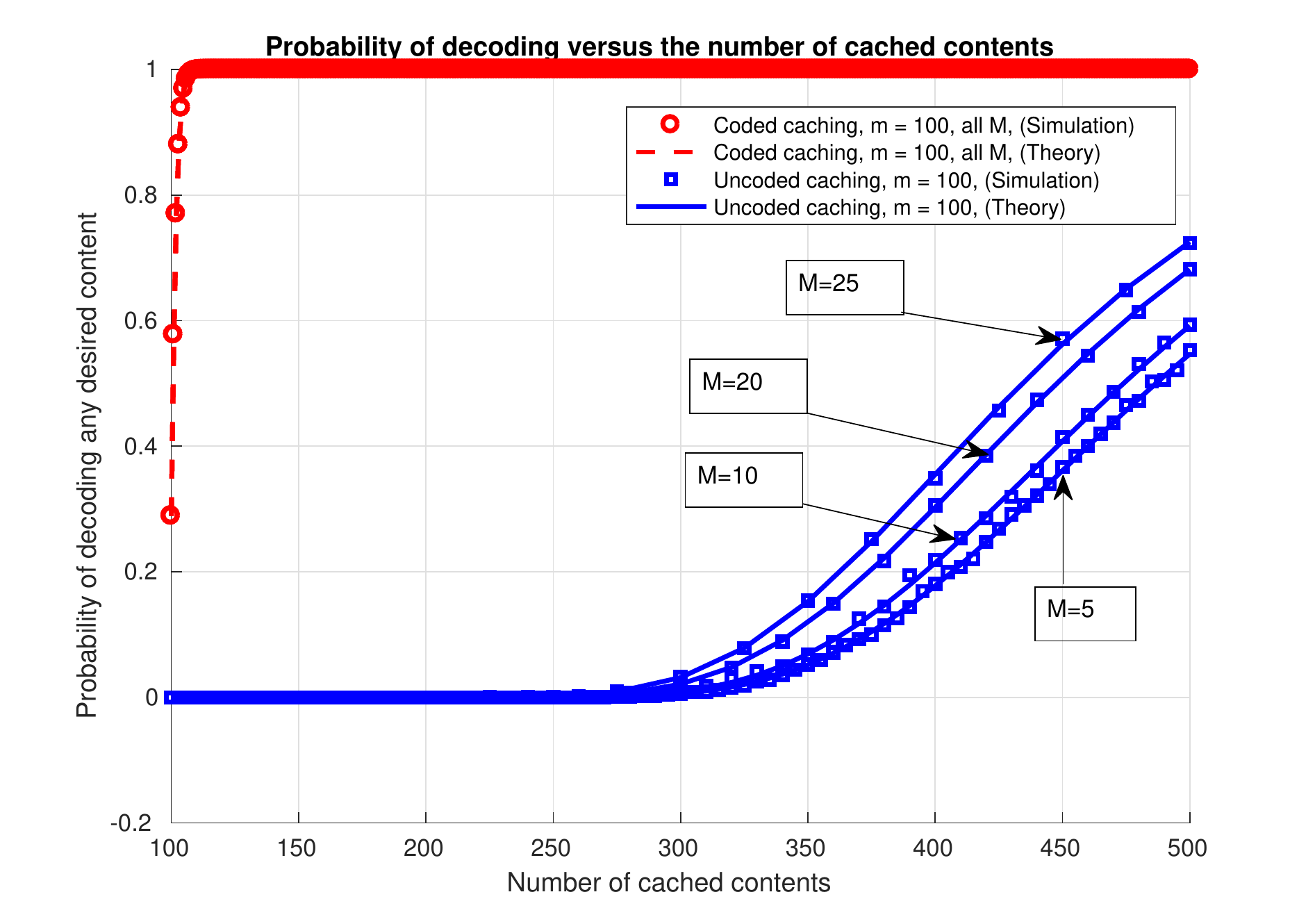}
\caption{Cache hit probability for any desired content when $m=100$.} 
\vspace{-0.2in}
\label{fig_m100}
\end{figure}

\section{Conclusions}
\label{conc_sec}
This paper introduces a novel decentralized coded caching strategy in wireless ad hoc networks. The capacity of this approach is compared with that of uncoded caching for proactive and reactive routing protocols. While with proactive routing protocol, uncoded caching outperforms coded caching, coded caching performs better with reactive routing protocol. Interestingly, it was shown that by choosing any random direction, close to optimum number of hops can be obtained to retrieve any content in coded caching.

It has been  shown that coded caching approach provides asymptotic perfect secrecy during file transmission. We have also studied the cache hit problem and shown that the cache hit probability for any desired content will be significantly higher in the proposed technique compared to uncoded caching. This technique reduces the problem of overcaching in the networks. An efficient and secure cache update algorithm is also proposed and the results are validated with simulations. 

The paper considered a static network, however our results can be extended to mobile or vehicular networks when the transmission time is much smaller than the network dynamics.
Further, the size of the file chunks can be adjusted based on the network dynamics to adapt to the rate of change in the network evolution. 

This paper considered random linear fountain codes to achieve asymptotic perfect secrecy and better cache hit probability and capacity results. 
 Other fountain coding choices for cache placement may reduce the decoding complexity while reducing the capacity and cache hit probability. Selection of appropriate fountain code depends on many factors such as  decoding complexity, delay and capacity requirements.

\bibliographystyle{plain}
\bibliography{All-Papers-J5}

\begin{IEEEbiography}
[{\includegraphics[width=1in,height=1.25in,clip,keepaspectratio]{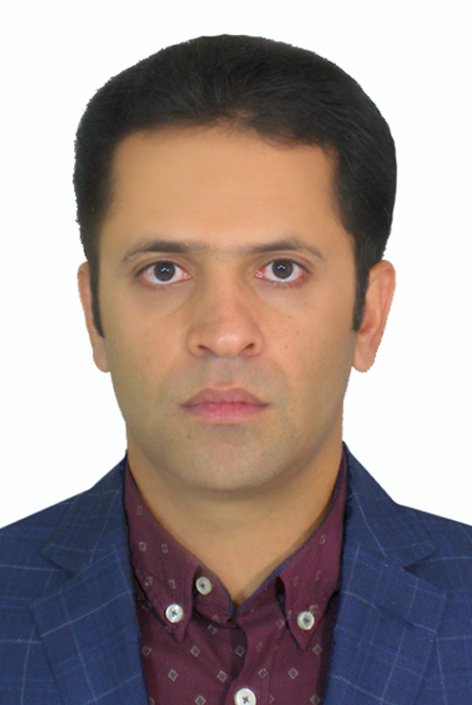}}] {Mohsen Karimzadeh Kiskani} received his B.S. degree in Mechanical Engineering and M.S. degree in Electrical Engineering from Sharif University of Technology in 2008 and 2010, respectively. He is currently a Ph.D. candidate in Electrical Engineering department at University of California, Santa Cruz. 
His  research interests  include wireless communications, information theory, and the application of fountain codes, LT codes and Raptor codes in wireless networks and storage systems.
He is also interested in  complexity study of Constraint Satisfaction Problems (CSP) in Computer Science. In June 2016 he obtained a M.S. degree in Computer Science from University of California Santa Cruz. 
\end{IEEEbiography}
\begin{IEEEbiography}
[{\includegraphics[width=1in,height=1in,clip,keepaspectratio]{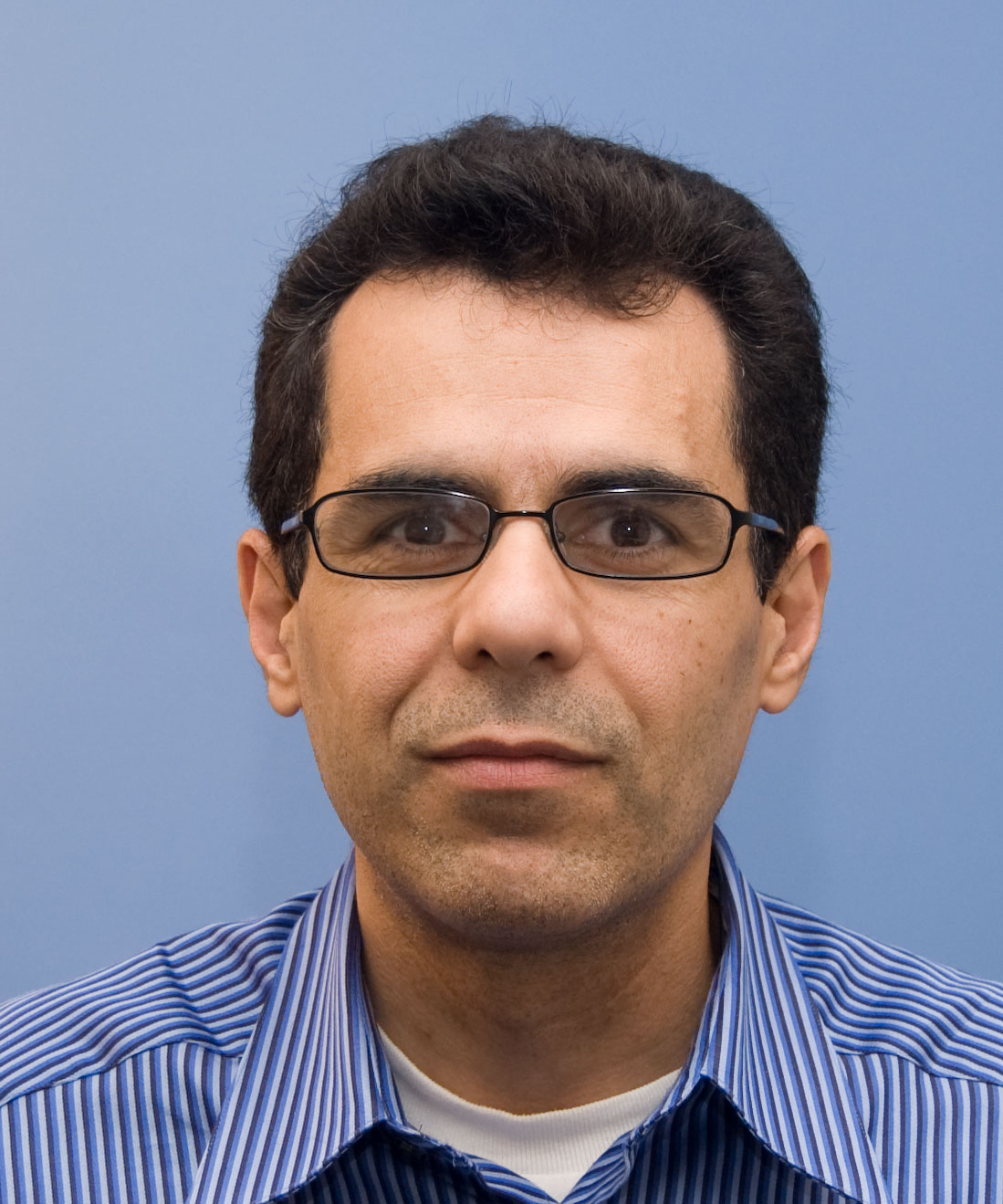}}] {Hamid Sadjadpour}
(S’94–M’95–SM’00) received Ph.D. degree in Electrical Engineering from  University of Southern California at Los Angeles, CA. In 1995, he joined
 AT\&T Research Laboratory in Florham Park, NJ as a Technical Staff Member and later as a Principal
Member of Technical Staff. In 2001, he joined  University of California at Santa Cruz, CA,
where he is currently a Professor. His research interests are in the general areas of wireless communications and networks. \end{IEEEbiography}

\end{document}